\documentclass{article}

\usepackage{arxiv}

\usepackage[utf8]{inputenc} 
\usepackage[T1]{fontenc}    
\usepackage{hyperref}       
\usepackage{url}            
\usepackage{booktabs}       
\usepackage{amsfonts}       
\usepackage{nicefrac}       
\usepackage{microtype}      
\usepackage{lipsum}		
\usepackage{graphicx}
\usepackage{epstopdf}
\usepackage{subfigure}
\usepackage{amsmath}
\usepackage{amsthm}

\newtheorem{theorem}{Theorem}[section]

\newtheorem{lemma}[theorem]{Lemma}

\newtheorem{remark}[theorem]{Remark}

\title{Boundary-free Estimators of the Mean Residual Life Function by Transformation}

\author{ \href{https://orcid.org/0000-0002-6939-9465}{\includegraphics[scale=0.06]{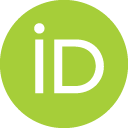}\hspace{1mm}Rizky Reza Fauzi}\thanks{corresponding author}\\
	Graduate School of Mathematics\\
	Kyushu University\\
	744 Motooka, Nishi-ku, Fukuoka-shi, Fukuoka-ken, JAPAN\\
	\texttt{fauzi.rizky.853@s.kyushu-u.ac.jp}\\
	\And
	Yoshihiko Maesono\\
	Faculty of Sciemce and Engineering\\
	Chuo University\\
	1-13-27 Kasuga, Bunkyo-ku, Tokyo-to, JAPAN\\
	\texttt{maesono@math.chuo-u.ac.jp}\\
}

\hypersetup{
pdftitle={A template for the arxiv style},
pdfsubject={q-bio.NC, q-bio.QM},
pdfauthor={David S.~Hippocampus, Elias D.~Striatum},
pdfkeywords={First keyword, Second keyword, More},
}

\begin{document}
\maketitle

\begin{abstract}
We propose two new kernel-type estimators of the mean residual life function $m_X(t)$ of bounded or half-bounded interval supported distributions. Though not as severe as the boundary problems in the kernel density estimation, eliminating the boundary bias problems that occur in the naive kernel estimator of the mean residual life function is needed. In this article, we utilize the property of bijective transformation. Furthermore, our proposed methods preserve the mean value property, which cannot be done by the naive kernel estimator. Some simulation results showing the estimators' performances and a real data analysis will be presented in the last part of this article.
\end{abstract}


\section{Introduction}\label{intro}
Statistical inference for remaining lifetimes would be intuitively more appealing than the popular hazard rate function, since its interpretation as ``the risk of immediate failure'' can be difficult to grasp. A function called the mean residual life (or mean excess loss) which represents ``the average remaining time before failure'' is easier to understand. The mean residual life (or MRL for short) function is of interest in many fields relating to time and finance, such as biomedical theory, survival analysis, and actuarial science.

Let $X_1,X_2,...,X_n$ be independently and identically distributed absolutely continuous random variables supported on an interval $\Omega\subset\mathbb{R}$, where $\inf\Omega=\omega'$, $\sup\Omega=\omega''$, and $-\infty\leq\omega'<\omega''\leq\infty$. Also, let $f_X(t)$ be the density function, $F_X(t)$ be the cumulative distribution function, $S_X(t)=\Pr(X>t)$ be the survival function, and $\mathbb{S}_X(t)=\int_t^\infty S_X(x)\mathrm{d}x$ be the cumulative survival function, of $X$. Then
\begin{gather}
m_X(t)=E(X-t|X>t), \; \; \; t\in\Omega
\end{gather}
is the definition of the mean residual life function, or can be written as
\begin{gather}
m_X(t)=\frac{\mathbb{S}_X(t)}{S_X(t)}.
\end{gather}
For a detailed discussion about the MRL function, see Embrechts \textit{et al.}~\cite{EKM97} or Guess and Proschan~\cite{GP88}. Murari and Sujit~\cite{MS95} and Belzunce \textit{et al.}~\cite{BRPS96} discussed the use of the MRL function for ordering and classifying distributions. On the other hand, Cox~\cite{cox62}, Kotz and Shanbhag~\cite{KS80}, and Zoroa \textit{et al.}~\cite{ZRM90} proposed how to determine distribution via an inversion formula of $m_X(t)$. Ruiz and Navarro~\cite{RN94} have considered the problem of characterization of the distribution function through the relationship between the MRL function and the hazard rate function. The MRL functions of finite mixtures and order statistics have been studied as well by Navarro and Hernandez~\cite{NH08}.

Some properties and applications of the MRL concept related to operational research and reliability theory in engineering are interesting topics. While Nanda \textit{et al.}~\cite{NBB10} discussed the properties of associated orderings in the MRL function, Huynh \textit{et al.}~\cite{HCBB14} studied the usefulness of the MRL models for maintenance decision-making.

Another examples are the utilization of the MRL functions of parallel system by Sadegh~\cite{sad08}, the MRL for records by Raqab and Asadi~\cite{RA08}, the MRL of a $k$-out-of-$n$:G system by Eryilmaz~\cite{ery12}, the MRL of a $(n-k+1)$-out-of-$n$ system by Poursaeed~\cite{pou10}, the MRL in reliability shock models by Eryilmaz~\cite{ery17}, the MRL subjected to Marshall-Olkin type shocks by Bayramoglu and Ozkut~\cite{BO16}, the MRL of coherent systems by Eryilmaz \textit{et al.}~\cite{ECC18} and Kavlak~\cite{kav17}, the MRL for degrading systems by Zhao \textit{et al.}~\cite{ZMCL18}, and the MRL of rail wagon bearings by Ghasemi and Hodkiewicz~\cite{GH12}.

The natural estimator of the MRL function is the empirical one, defined as
\begin{gather}
m_n(t)=\frac{\mathbb{S}_n(t)}{S_n(t)}=\frac{\sum_{i=1}^n(X_i-t)I(X_i>t)}{\sum_{i=1}^n I(X_i>t)}, \; \; \; t\in\Omega,
\end{gather}
where $I(A)$ is the usual indicator function on set $A$. Yang~\cite{yan78}, Ebrahimi~\cite{ebr91}, and Cs\"{o}rg\H{o} and Zitikis~\cite{CZ96} studied the properties of $m_n(t)$. Even though it has several good attributes (e.g. unbiasedness and consistency), the empirical MRL function is just a rough estimate of $m_X(t)$ and lack of smoothness. Estimating is also impossible for large $t$ because $S_n(t)=0$ for $t>\max\{X_1,X_2,...,X_n\}$. Though we can just define $m_n(t)=0$ for such case, it is a major disadvantage as analysing the behaviour of the MRL function when $t\rightarrow\infty$ is of an interest.

Various parametric models of MRL have been discussed in literatures, for example the transformed parametric MRL models by Sun and Zhang~\cite{SZ09}, the upside-down bathtub-shaped MRL model by Shen \textit{et al.}~\cite{STX09}, the MRL order of convolutions of heterogeneous exponential random variables by Zhao and Balakrishnan~\cite{ZB09}, the proportional MRL model by Nanda \textit{et al.}~\cite{NBA06} and Chan \textit{et al.}~\cite{CCD12}, and the MRL models with time-dependent coefficients by Sun \textit{et al.}~\cite{SSZ12}.

Some nonparametric estimators of $m_X(t)$ which are related to the empirical one have been discussed in a fair amount of literature. For example, Ruiz and Guillam\'{o}n~\cite{RG96} estimated the numerator in $m_n(t)$ by a recursive kernel estimate and left the empirical survival function unchanged, while Chaubey and Sen~\cite{CS99} used the Hille's Theorem in Hille~\cite{hil48} to smooth both the numerator and denominator in $m_n(t)$.

The other maneuver that can be used for estimating the MRL function nonparametrically is the kernel method. Let $K(x)$ be a symmetric continuous nonnegative kernel function with $\int_{-\infty}^\infty K(x)\mathrm{d}x=1$, and $h>0$ be a bandwidth satisfying $h\rightarrow 0$ and $nh\rightarrow\infty$ when $n\rightarrow\infty$. From this, we will have three other functions derived from $K(x)$, they are
\begin{eqnarray}
W(x)=\int_{-\infty}^x K(z)\mathrm{d}z, \; \; \; V(x)=\int_x^\infty K(z)\mathrm{d}z, \; \; \; \mathrm{and} \; \; \; \mathbb{V}(x)=\int_x^\infty V(z)\mathrm{d}z.
\end{eqnarray}
Hence, the naive kernel MRL function estimator can be defined as
\begin{gather}
\widehat{m}_X(t)=\frac{\widehat{\mathbb{S}}_X(t)}{\widehat{S}_X(t)}=\frac{h\sum_{i=1}^n\mathbb{V}\left(\frac{t-X_i}{h}\right)}{\sum_{i=1}^n V\left(\frac{t-X_i}{h}\right)}, \; \; \; t\in\Omega.
\end{gather}
Guillam\'{o}n \textit{et al.}~\cite{GNR96} discussed the asymptotic properties of the naive kernel MRL function estimator in detail.

However, as usually $m_X(t)$ is used for time or finance related data, which are on nonnegative real line or bounded interval, the naive kernel MRL function estimator suffers the so called boundary bias problem. In the case of $f_X(\omega')=0$ (or $f_X(\omega'')=0$), the boundary effects of $\widehat{m}_X(t)$ when $t\rightarrow\omega'$ (or $t\rightarrow\omega''$) is not as bad as in the kernel density estimator, but the problems still occur. It is because the term $S_X(\omega')$ and $1-S_X(\omega'')$ in the $Bias[\widehat{\mathbb{S}}_X(\omega')]$ and $Bias[\widehat{\mathbb{S}}_X(\omega'')]$ can never be $0$ since $S_X(\omega')=1-S_X(\omega'')=1$, which means $\widehat{\mathbb{S}}_X(t)$ causes the boundary problems for $\widehat{m}_X(t)$. Moreover, in the case of $f_X(\omega')>0$ and $f_X(\omega'')>0$ (e.g. uniform distribution), not only $\widehat{\mathbb{S}}_X(t)$, but $\widehat{S}_X(t)$ also adds its share to the boundary problems for $\widehat{m}_X(t)$.

To make things worse, the naive kernel MRL function estimator does not preserve one of the most important properties of the MRL function, which is $m_X(\omega')+\omega'=E(X)$. It is reasonable if we expect $\widehat{m}_X(\omega')+\omega'\approx\bar{X}$. However, $\widehat{S}_X(\omega')$ is less than $1$ and $\widehat{\mathbb{S}}_X(\omega')$ is smaller than the average value of $X_i's$, due to the weight that they still put on the outside of $\Omega$. Accordingly, there is no guarantee of how far or how close $\widehat{m}_X(\omega')+\omega'$ is to $\bar{X}$. Some simulations in Section~\ref{sec:sim} illustrate this statement.

Some articles have suggested methods to solve the boundary bias problems in the density estimation, such as data reflection by Schuster~\cite{sch85}; simple nonnegative boundary correction by Jones and Foster~\cite{JF96}; boundary kernels by M\"{u}ller~\cite{mul91}, M\"{u}ller~\cite{mul93}, and M\"{u}ller and Wang~\cite{MW94}; generating pseudo data by Cowling and Hall~\cite{CH96}; hybrid method by Hall and Wehrly~\cite{HW91}; and the local polynomial fitting by Fan and Gijbels~\cite{FG96}. Even though only few literatures discussed how to extend previous ideas for solving the problems in the MRL function estimation, Abdous and Berred~\cite{AB05} successfully adopted the idea of local polynomial fitting (linear in their case) for the MRL function estimation.

In this article we are going to try another idea to remove the boundary effects, which is utilizing transformations that map $\Omega$ to $\mathbb{R}$ bijectively. In this situation there are no boundary effects at all, as we will not put any weight outside the support. Hence, instead of using $X_1,X_2,...,X_n$, we will apply the kernel method for the transformed $Y_1,Y_2,...,Y_n$, where $Y_i=g^{-1}(X_i)$ and $g:\mathbb{R}\rightarrow\Omega$ is a bijective function. However, even though the idea is easy to understand, we cannot just substitute $t$ with $g^{-1}(t)$ and $X_i$ with $Y_i$ in the formula of $\widehat{m}_X(t)$, due to avoiding nonintegrability. We need to modify the naive kernel MRL function estimator before substituting $g^{-1}(t)$ and $Y_i$ in order to preserve the integrability and to ensure that the new formulas are good estimators of the mean excess loss function.

Before moving on to our main focus, we need to impose some conditions:
\begin{enumerate}
	\item[\textbf{C1}.] The kernel $K(x)$ is a continuous nonnegative function and symmetric at $x=0$ with $\int_{-\infty}^\infty K(x)\mathrm{d}x=1$
	\item[\textbf{C2}.] The bandwidth $h>0$ satisfies $h\rightarrow 0$ and $nh\rightarrow\infty$ when $n\rightarrow\infty$
	\item[\textbf{C3}.] The function $g:\mathbb{R}\rightarrow\Omega$ is continuous and strictly increasing
	\item[\textbf{C4}.] The density $f_X$ and the function $g$ are continuously differentiable at least twice
	\item[\textbf{C5}.] The integrals $\int_{-\infty}^\infty g'(ux)K(x)\mathrm{d}x$ and $\int_{-\infty}^\infty g'(ux)V(x)\mathrm{d}x$ are finite for all $u$ in an $\varepsilon$-neighbourhood of the origin
	\item[\textbf{C6}.] The expectations $E(X)$, $E(X^2)$, and $E(X^3)$ exist.
\end{enumerate}
The first and the second conditions are standard assumptions for kernel methods, and C3 is needed for the bijectivity and the simplicity of the transformation. Since we will use some expansions of the survival and the cumulative survival functions, C4 is important to ensure the validity of our proofs. The last two conditions are necessary to make sure we can derive the bias and the variance formulas. In order to calculate the variances, we also define a new function
\begin{gather}
\bar{\mathbb{S}}_X(t)=\int_t^\infty\mathbb{S}_X(x)\mathrm{d}x
\end{gather}
for simpler notation. Some numerical studies are discussed in Section~\ref{sec:num}, and the detailed proofs can be found in the appendices.

\section{Estimators of the survival function and the cumulative survival function}
\label{sec:sfcsf}
Before jumping into the estimation of the mean residual life function, we will first discuss on the estimations of each component, which are the survival function $S_X(t)$ and the cumulative survival function $\mathbb{S}_X(t)$. In this article, we proposed two sets of estimators using the idea of transformation. Based on those two sets of estimators, we will propose two estimators of the MRL function in Section~\ref{sec:mrlf}.

Geenens~\cite{gee14}, also Wen and Wu~\cite{WW15}, used probit transformation to eliminate the boundary bias problems in the kernel density estimation for data on the unit interval. If we generalize their idea for any interval $\Omega$ and using any function $g$ that satisfies the conditions stated before, then we will have
\begin{gather*}
\widetilde{f}_X(t)=\frac{1}{nhg'(g^{-1}(t))}\sum_{i=1}^n K\left(\frac{g^{-1}(t)-g^{-1}(X_i)}{h}\right)
\end{gather*}
as the generalized boundary-free kernel density estimator by transformation. Then, by doing simple subtitution technique on $\int_{t}^{\omega''}\widetilde{f}_X(x)\mathrm{d}x$, the first proposed survival function estimator is
\begin{gather}
\widetilde{S}_{X,1}(t)=\frac{1}{n}\sum_{i=1}^n V_{1,h}(g^{-1}(t),g^{-1}(X_i)), \; \; \; t\in\Omega,
\end{gather}
where
\begin{gather}
V_{1,h}(x,y)=\frac{1}{h}\int_x^\infty K\left(\frac{z-y}{h}\right)\mathrm{d}z.
\end{gather}
Using the same approach, we define the first proposed cumulative survival function estimator as
\begin{gather}
\widetilde{\mathbb{S}}_{X,1}(t)=\frac{1}{n}\sum_{i=1}^n\mathbb{V}_{1,h}(g^{-1}(t),g^{-1}(X_i)), \; \; \; t\in\Omega,
\end{gather}
where
\begin{gather}
\mathbb{V}_{1,h}(x,y)=\int_x^\infty g'(z)V\left(\frac{z-y}{h}\right)\mathrm{d}z.
\end{gather}
Their biases and variances are given in the following theorem.
\begin{theorem}\label{thm:set1}
	Under the condition C1-C6, the biases and the variances of $\widetilde{S}_{X,1}(t)$ and $\widetilde{\mathbb{S}}_{X,1}(t)$ are
	\begin{gather}
	Bias[\widetilde{S}_{X,1}(t)]=-\frac{h^2}{2}b_1(t)\int_{-\infty}^\infty y^2 K(y)\mathrm{d}y+o(h^2),\\
	Var[\widetilde{S}_{X,1}(t)]=\frac{1}{n}S_X(t)F_X(t)-\frac{h}{n}g'(g^{-1}(t))f_X(t)\int_{-\infty}^\infty V(y)W(y)\mathrm{d}y+o\left(\frac{h}{n}\right),
	\end{gather}
	and
	\begin{gather}
	Bias[\widetilde{\mathbb{S}}_{X,1}(t)]=\frac{h^2}{2}b_2(t)\int_{-\infty}^\infty y^2 K(y)\mathrm{d}y+o(h^2),\\
	Var[\widetilde{\mathbb{S}}_{X,1}(t)]=\frac{1}{n}[2\bar{\mathbb{S}}_X(t)-\mathbb{S}_X^2(t)]+o\left(\frac{h}{n}\right),
	\end{gather}
	where
	\begin{gather}
	b_1(t)=g''(g^{-1}(t))f_X(t)+[g'(g^{-1}(t))]^2 f_X'(t),\\
	b_2(t)=[g'(g^{-1}(t))]^2 f_X(t)+\int_{g^{-1}(t)}^\infty g''(x)g'(x)f_X(g(x))\mathrm{d}x.
	\end{gather}
	Furthermore, the covariance of them is
	\begin{gather}
	Cov[\widetilde{\mathbb{S}}_{X,1}(t),\widetilde{S}_{X,1}(t)]=\frac{1}{n}\mathbb{S}_X(t)F_X(t)+o\left(\frac{h}{n}\right).
	\end{gather}
\end{theorem}
\begin{remark}\label{rem:set1}
	Because $\frac{\mathrm{d}}{\mathrm{d}t}\widetilde{\mathbb{S}}_{X,1}(t)=-\widetilde{S}_{X,1}(t)$, it means that our first set of estimators preserves the relationship between the theoretical $\mathbb{S}_X(t)$ and $S_X(t)$.
\end{remark}

We have utilized the relationship among density, survival, and cumulative survival functions to construct the first set of estimators, now we are going to use another maneuver to build our second set of estimators. The second proposed survival function estimator is defined as
\begin{gather}
\widetilde{S}_{X,2}(t)=\frac{1}{n}\sum_{i=1}^n V_{2,h}(g^{-1}(t),g^{-1}(X_i)), \; \; \; t\in\Omega,
\end{gather}
where
\begin{gather}
V_{2,h}(x,y)=V\left(\frac{x-z}{h}\right).
\end{gather}
As we can see, $\widetilde{S}_{X,2}(t)$ is basically just a result of a simple subsitution of $g^{-1}(t)$ and $g^{-1}(X_i)$ to the formula of $\widehat{S}_X(t)$. This can be done due to the change-of-variable property of the survival function (for a brief explanation of the change-of-variable property, see Lemma~\ref{lemma:cv}). Though it is bit trickier, the change-of-variable property of the cumulative survival function leads us to the construction of our second proposed cumulative survival function estimator, which is
\begin{gather}
\widetilde{\mathbb{S}}_{X,2}(t)=\frac{1}{n}\sum_{i=1}^n\mathbb{V}_{2,h}(g^{-1}(t),g^{-1}(X_i)), \; \; \; t\in\Omega,
\end{gather}
where
\begin{gather}
\mathbb{V}_{2,h}(x,y)=\int_{-\infty}^y g'(z)V\left(\frac{x-z}{h}\right)\mathrm{d}z.
\end{gather}
In the above formula, multiplying $V$ with $g'$ is necessary to make sure that $\widetilde{\mathbb{S}}_{X,2}(t)$ is an estimator of $\mathbb{S}_X(t)$ (see equation (\ref{eq:a6})). Now, with $\widetilde{\mathbb{S}}_{X,2}(t)$ and $\widetilde{S}_{X,2}(t)$, their biases and variances are as follows.
\begin{theorem}\label{thm:set2}
	Under the condition C1-C6, the biases and the variances of $\widetilde{S}_{X,2}(t)$ and $\widetilde{\mathbb{S}}_{X,2}(t)$ are
	\begin{gather}
	Bias[\widetilde{S}_{X,2}(t)]=-\frac{h^2}{2}b_1(t)\int_{-\infty}^\infty y^2 K(y)\mathrm{d}y+o(h^2),\\
	Var[\widetilde{S}_{X,2}(t)]=\frac{1}{n}S_X(t)F_X(t)-\frac{h}{n}g'(g^{-1}(t))f_X(t)\int_{-\infty}^\infty V(y)W(y)\mathrm{d}y+o\left(\frac{h}{n}\right),
	\end{gather}
	and
	\begin{gather}
	Bias[\widetilde{\mathbb{S}}_{X,2}(t)]=\frac{h^2}{2}b_3(t)\int_{-\infty}^\infty y^2 K(y)\mathrm{d}y+o(h^2),\\
	Var[\widetilde{\mathbb{S}}_{X,2}(t)]=\frac{1}{n}[2\bar{\mathbb{S}}_X(t)-\mathbb{S}_X^2(t)]+o\left(\frac{h}{n}\right),
	\end{gather}
	where
	\begin{gather}
	b_3(t)=[g'(g^{-1}(t))]^2 f_X(t)-g''(g^{-1}(t))S_X(t).
	\end{gather}
	Furthermore, the covariance of them is
	\begin{gather}
	Cov[\widetilde{\mathbb{S}}_{X,2}(t),\widetilde{S}_{X,2}(t)]=\frac{1}{n}\mathbb{S}_X(t)F_X(t)+o\left(\frac{h}{n}\right).
	\end{gather}
\end{theorem}
\begin{remark}\label{rem:set2a}
	As we can see in Theorem~\ref{thm:set1} and Theorem~\ref{thm:set2}, a lot of similarities are possessed by both sets of estimators. For example, both of them have the same covariances, which means the statistical relationship between $\widetilde{S}_{X,2}(t)$ and $\widetilde{\mathbb{S}}_{X,2}(t)$ is same to the one of $\widetilde{S}_{X,1}(t)$ and $\widetilde{\mathbb{S}}_{X,1}(t)$.
\end{remark}
\begin{remark}\label{rem:set2b}
	We can prove that both $\widetilde{S}_{X,1}(\omega')$ and $\widetilde{S}_{X,2}(\omega')$ are always equal to $1$ (see Appendix G), and it is obvious that both $\widetilde{S}_{X,1}(\omega'')$ and $\widetilde{S}_{X,2}(\omega'')$ are $0$. Hence, it is clear that their variances are $0$ when $t$ approaches the boundaries. This is one of the reasons our proposed methods outperform the naive kernel estimator.
\end{remark}

\section{Estimators of the mean residual life function}
\label{sec:mrlf}
In this section, we will discuss the estimation for the mean residual life function. As we already have defined the survival function and the cumulative survival function estimators, we just need to plug them into the MRL function formula. Hence, our proposed estimators of the mean excess loss function are
\begin{gather}
\widetilde{m}_{X,1}(t)=\frac{\widetilde{\mathbb{S}}_{X,1}(t)}{\widetilde{S}_{X,1}(t)}=\frac{h\sum_{i=1}^n\int_{g^{-1}(t)}^\infty g'(z)V\left(\frac{z-g^{-1}(X_i)}{h}\right)\mathrm{d}z}{\sum_{i=1}^n\int_{g^{-1}(t)}^\infty K\left(\frac{z-g^{-1}(X_i)}{h}\right)\mathrm{d}z}, \; \; \; t\in\Omega,
\end{gather}
and
\begin{gather}
\widetilde{m}_{X,2}(t)=\frac{\widetilde{\mathbb{S}}_{X,2}(t)}{\widetilde{S}_{X,2}(t)}=\frac{\sum_{i=1}^n\int_{-\infty}^{g^{-1}(X_i)}g'(z)V\left(\frac{g^{-1}(t)-z}{h}\right)\mathrm{d}z}{\sum_{i=1}^n V\left(\frac{g^{-1}(t)-g^{-1}(X_i)}{h}\right)}, \; \; \; t\in\Omega.
\end{gather}
At first glance, $\widetilde{m}_{X,1}(t)$ seems more representative to the theoretical $m_X(t)$, since the mathematical relationship between $\widetilde{S}_{X,1}(t)$ and $\widetilde{\mathbb{S}}_{X,1}(t)$ are same as the relationship between the numerator and the denumerator of $m_X(t)$, as stated in Remark~\ref{rem:set1}. This is not a major problem for $\widetilde{m}_{X,2}(t)$, as we stated in Remark~\ref{rem:set2a} that the relationship between $\widetilde{S}_{X,2}(t)$ and $\widetilde{\mathbb{S}}_{X,2}(t)$ is statistically same to the relationship between $\widetilde{S}_{X,1}(t)$ and $\widetilde{\mathbb{S}}_{X,1}(t)$. However, when a statistician wants to keep the mathematical relationship between the survival and the cumulative survival functions in their estimates, it is suggested to use $\widetilde{m}_{X,1}(t)$ instead.
\begin{theorem}\label{thm:bmrlf}
	Under the condition C1-C6, the biases and the variances of $\widetilde{m}_{X,i}(t)$, $i=1,2$, are
	\begin{gather}
	Bias[\widetilde{m}_{X,1}(t)]=\frac{h^2}{2S_X(t)}[b_2(t)+m_X(t)b_1(t)]\int_{-\infty}^\infty y^2 K(y)\mathrm{d}y+o(h^2),\\
	Bias[\widetilde{m}_{X,2}(t)]=\frac{h^2}{2S_X(t)}[b_3(t)+m_X(t)b_1(t)]\int_{-\infty}^\infty y^2 K(y)\mathrm{d}y+o(h^2),\\
	Var[\widetilde{m}_{X,i}(t)]=\frac{1}{n}\frac{b_4(t)}{S_X^2(t)}-\frac{h}{n}\frac{b_5(t)}{S_X^2(t)}\int_{-\infty}^\infty V(y)W(y)\mathrm{d}y+o\left(\frac{h}{n}\right),
	\end{gather}
	where
	\begin{gather}
	b_4(t)=2\bar{\mathbb{S}}_X(t)-S_X(t)m_X^2(t) \; \; \;  \mathrm{and} \; \; \; b_5(t)=g'(g^{-1}(t))f_X(t)m_X^2(t).
	\end{gather}
\end{theorem}

Similar to most of kernel type estimators, our proposed estimators attain asymptotic normality, as stated in Theorem~\ref{thm:normal}.
\begin{theorem}\label{thm:normal}
	Under the condition C1-C6, the limiting distribution
	\begin{gather}
	\frac{\widetilde{m}_{X,i}(t)-m_X(t)}{\sqrt{Var[\widetilde{m}_{X,i}(t)]}}\rightarrow_D N(0,1)
	\end{gather}
	holds for $i=1,2$.
\end{theorem}
Furthermore, we also establish strong consistency of the proposed estimators in the form of the following theorem.
\begin{theorem}\label{thm:consistent}
	Under the condition C1-C6, the consistency
	\begin{gather}
	\sup_{t\in\Omega}|\widetilde{m}_{X,i}(t)-m_X(t)|\rightarrow_{a.s.}0
	\end{gather}
	holds for $i=1,2$.
\end{theorem}

The last property that we would like to discuss is the behaviour of our proposed estimators when $t$ is in the boundary regions. As stated in Section~\ref{intro}, we want our estimators to preserve the behaviour of the theoretical MRL function, specifically the property of $m_X(\omega')=E(X)-\omega'$. If we can prove this, then not only will our proposed methods be free of boundary problems, but also superior in the sense of them preserving the key property of the MRL function.
\begin{theorem}\label{thm:preserve}
	Let $\widetilde{m}_{X,1}(t)$ and $\widetilde{m}_{X,2}(t)$ be the transformed kernel mean residual life function estimators. Then
	\begin{gather}
	\widetilde{m}_{X,1}(\omega')+\omega'=\bar{X}+O_p(h^2),
	\end{gather}
	and
	\begin{gather}\label{eq:unbiased}
	\widetilde{m}_{X,2}(\omega')+\omega'=\bar{X}.
	\end{gather}
\end{theorem}
\begin{remark}\label{rem:preservea}
	Please note that, although for convenience it is written as $\widetilde{m}_{X,i}(\omega')$ (or $\widetilde{m}_{X,i}(\omega'')$), but we actually mean it as $\lim_{t\rightarrow\omega'^+}\widetilde{m}_{X,i}(t)$ (or $\lim_{t\rightarrow\omega''^-}\widetilde{m}_{X,i}(t)$), since $g^{-1}(\omega')$ (or $g^{-1}(\omega'')$) might be undefined.
\end{remark}
\begin{remark}\label{rem:preserveb}
	From equation (\ref{eq:unbiased}), we can say that $\widetilde{m}_{X,2}(\omega')$ is unbiased, because
	\begin{gather*}
	E[\widetilde{m}_{X,2}(\omega')]=E(X)-\omega'=m_X(\omega').
	\end{gather*}
	In other words, its bias is exactly $0$. On the other hand, even though $\widetilde{m}_{X,1}(\omega')$ is not exactly the same as $\bar{X}-\omega'$, we can at least say they are close enough, and the rate of $h^2$ error is relatively small. However, from this we may take a conclusion that $\widetilde{m}_{X,2}(t)$ is superior than $\widetilde{m}_{X,1}(t)$ in the aspect of preserving behaviour of the MRL function near the boundary.
\end{remark}

\section{Numerical studies}
\label{sec:num}
In this section, we show the results of our numerical studies. The studies are divided into two parts, the simulations and the real data analysis.

\subsection{simulation results}
\label{sec:sim}
In this study, we calculated the average integrated squared error (AISE) and the average squared error (ASE) with several sample sizes, and repeated them $10000$ times for each case. We compared four estimators: empirical $m_n(t)$; naive kernel $\widehat{m}_X(t)$; and our two proposed estimators $\widetilde{m}_{X,1}(t)$ and $\widetilde{m}_{X,2}(t)$. The distributions which we generated are standard uniform $U(0,1)$, beta $Beta(3,2)$, gamma $Gamma(2,3)$, Weibull $Weibull(3,2)$, and absolute-normal $abs.N(0,1)$ distributions. For $U(0,1)$ and $Beta(3,2)$, we took $g=\Phi$, the standard normal distribution function; and we chose $g^{-1}=\log$ for the rests. The kernel function we used here is the Epanechnikov Kernel and the bandwidths were chosen by cross-validation technique. We actually did the same simulation study using Gaussian Kernel, but the results are quite similar. That being the case, we do not show those results in this article.

(Table~\ref{tab:aise}) compares the AISE in order to illustrate the general measure of error among the estimators. (Table~\ref{tab:ase1}) compares the ASE of each estimate when $t=0.001$, as a representation of the error when $t$ is in the boundary region. For (Table~\ref{tab:ase2}), the ASE at $t=E(X)$ represents the error when the point of evaluation is moderate. The last table represents the error of the estimators when $t$ is large enough.

As we can see in the tables, our proposed estimators gave the best results for all cases. This is particularly true for our second proposed estimator in most cases. Though our first proposed estimator's performances are not as good as the second one, it is still fairly comparable because the differences are not huge. Furthermore, the first proposed estimator is better than the empirical and the naive kernel estimators in most cases.

We may take interest in (Table~\ref{tab:ase1}) as the empirical MRL function gave similar results as our second proposed estimator did. However, this is reasonable due to the fact that $m_n(0)=\bar{X}$, same as $\widetilde{m}_{X,2}(0)$ according to Theorem~\ref{thm:preserve}. In (Table~\ref{tab:ase2}), even though our second estimator still outperformed the others, the margins of difference with the other estimators are not big. This can be explained as $t=E(X)$ has high density, neither it has boundary problems nor lack of data as in the tail. However, (Table~\ref{tab:ase3}) showed another story. As the tail of distribution has lesser density of data, the empirical and naive kernel estimators dropped to $0$ quickly. This explains why the ASE of theirs are much larger than the ASE of both of the proposed estimators.

\begin{table}
	\centering
	\caption{Average integrated squared error comparison}
	\label{tab:aise}
	\begin{tabular}{lllll}
		\hline\noalign{\smallskip}
		Distributions&Empirical&Naive&Proposed 1&Proposed 2\\
		\noalign{\smallskip}\hline\noalign{\smallskip}
		$U(0,1)$&$12.62201$&$12.92260$&$6.43667$&$\mathbf{6.06429}$\\
		$Beta(3,2)$&$31.49324$&$35.09813$&$14.27382$&$\mathbf{10.09131}$\\
		$Gamma(2,3)$&$58.86433$&$66.06180$&$29.01258$&$\mathbf{24.51671}$\\
		$Weibull(3,2)$&$0.14887$&$0.27073$&$0.11483$&$\mathbf{0.04073}$\\
		$abs.N(0,1)$&$0.13255$&$0.09379$&$0.07700$&$\mathbf{0.04943}$\\
		\noalign{\smallskip}\hline
	\end{tabular}
\end{table}

\begin{table}
	\centering
	\caption{Average squared error comparison when $t=0.001$}
	\label{tab:ase1}
	\begin{tabular}{lllll}
		\hline\noalign{\smallskip}
		Distributions&Empirical&Naive&Proposed 1&Proposed 2\\
		\noalign{\smallskip}\hline\noalign{\smallskip}
		$U(0,1)$&$\mathbf{0.08082}$&$0.27639$&$0.08728$&$\mathbf{0.08082}$\\
		$Beta(3,2)$&$0.23834$&$0.77734$&$0.28031$&$0.23833$\\
		$Gamma(2,3)$&$\mathbf{0.35739}$&$0.82174$&$0.37886$&$\mathbf{0.35739}$\\
		$Weibull(3,2)$&$\mathbf{0.00741}$&$0.05175$&$0.00958$&$\mathbf{0.00741}$\\
		$abs.N(0,1)$&$\mathbf{0.00728}$&$0.07433$&$0.00803$&$\mathbf{0.00728}$\\
		\noalign{\smallskip}\hline
	\end{tabular}
\end{table}

\begin{table}
	\centering
	\caption{Average squared error comparison when $t=E(X)$}
	\label{tab:ase2}
	\begin{tabular}{lllll}
		\hline\noalign{\smallskip}
		Distributions&Empirical&Naive&Proposed 1&Proposed 2\\
		\noalign{\smallskip}\hline\noalign{\smallskip}
		$U(0,1)$&$0.20656$&$0.26846$&$0.20237$&$\mathbf{0.18102}$\\
		$Beta(3,2)$&$0.42396$&$0.61622$&$0.42069$&$\mathbf{0.31329}$\\
		$Gamma(2,3)$&$0.76834$&$1.01913$&$0.73909$&$\mathbf{0.56705}$\\
		$Weibull(3,2)$&$0.00806$&$0.02740$&$0.03183$&$\mathbf{0.00665}$\\
		$abs.N(0,1)$&$0.01176$&$0.02295$&$0.01127$&$\mathbf{0.00566}$\\
		\noalign{\smallskip}\hline
	\end{tabular}
\end{table}

\begin{table}
	\centering
	\caption{Average squared error comparison when $t=E(X)+3\sigma$}
	\label{tab:ase3}
	\begin{tabular}{lllll}
		\hline\noalign{\smallskip}
		Distributions&Empirical&Naive&Proposed 1&Proposed 2\\
		\noalign{\smallskip}\hline\noalign{\smallskip}
		$U(0,1)$&$3.09158$&$3.06731$&$\mathbf{1.33397}$&$1.34459$\\
		$Beta(3,2)$&$5.89501$&$6.03789$&$1.85105$&$\mathbf{1.51738}$\\
		$Gamma(2,3)$&$10.61608$&$11.18569$&$4.26919$&$\mathbf{4.00378}$\\
		$Weibull(3,2)$&$0.13199$&$0.12829$&$0.06907$&$\mathbf{0.01307}$\\
		$abs.N(0,1)$&$0.08860$&$0.08057$&$0.03929$&$\mathbf{0.02293}$\\
		\noalign{\smallskip}\hline
	\end{tabular}
\end{table}

As further illustrations, we also provide some graphs to compare our proposed estimators' performances with the other estimators. (Figure~\ref{fig:mrl}) is about the graphs comparison of the empirical, the naive kernel, and our two proposed estimators. By (Figure~\ref{fig:bias}), we compare the point-wise simulated bias of the same estimators. From those, we can say that our proposed estimators outperformed the empirical and the naive kernel estimators.

There are three things that we want to emphasize from these figures. First, instead of resembling the theoretical shape, the graphs of the naive kernel estimator are more like a smoothed version of the graphs of the empirical estimator, especially in (Figure~\ref{fig:mrlexp}) and (Figure~\ref{fig:mrlgam}). This is somewhat interesting, as even though lack of smoothness, empirical type estimators (e.g. empirical distribution function) usually quite resemble the shape of the theoretical ones. However, in this MRL function case, the empirical MRL function cannot be used as a reference, because its shape is too unstable and different to the theoretical shape (see (Figure~\ref{fig:mrlexp}), (Figure~\ref{fig:mrlgam}), and (Figure~\ref{fig:mrlabs})). Same goes for the naive kernel MRL function estimator. Even though (Figure~\ref{fig:mrlwei}) and (Figure~\ref{fig:mrlabs}) showed the naive kernel estimator has nice graphs, it performed fairly poorly in (Figure~\ref{fig:mrlexp}) and (Figure~\ref{fig:mrlgam}). On the other hand, the graphs of our proposed estimators resemble the theoretical ones. The difference is quite striking in (Figure~\ref{fig:mrlexp}), where the empirical and naive kernel estimators are jumpy, but the proposed estimators gave stable and almost straight-line graphs.

The second thing we want to emphasize is, from all figures we can see that the boundary bias problems affect naive kernel estimator severely, as (Figure~\ref{fig:bias}) shows the simulated bias values of $\widehat{m}_X(t)$ near $t=0$ are the farthest from $0$. We can also conclude that the empirical MRL function does not suffer from the boundary bias problems, as its bias is almost $0$ near $t=0$. However, as $t$ goes larger, the bias drops to negative value quickly, especially in (Figure~\ref{fig:biasexp}). In contrast, our estimators, especially the second one, gave almost straight line at $0$-ordinate in (Figure~\ref{fig:biasabs}), which means its simulated bias is almost always $0$. And at last, we can conclude that though all of the graphs of the estimators presented here will fade to $0$ when $t$ is large enough, our proposed estimators are more stable and fading to $0$ much slower than the other two estimators.

\begin{figure}
	\centering
	\subfigure[$X\sim exp\left(\frac{1}{2}\right)$\label{fig:mrlexp}]{
		\resizebox*{5.85cm}{!}{\includegraphics{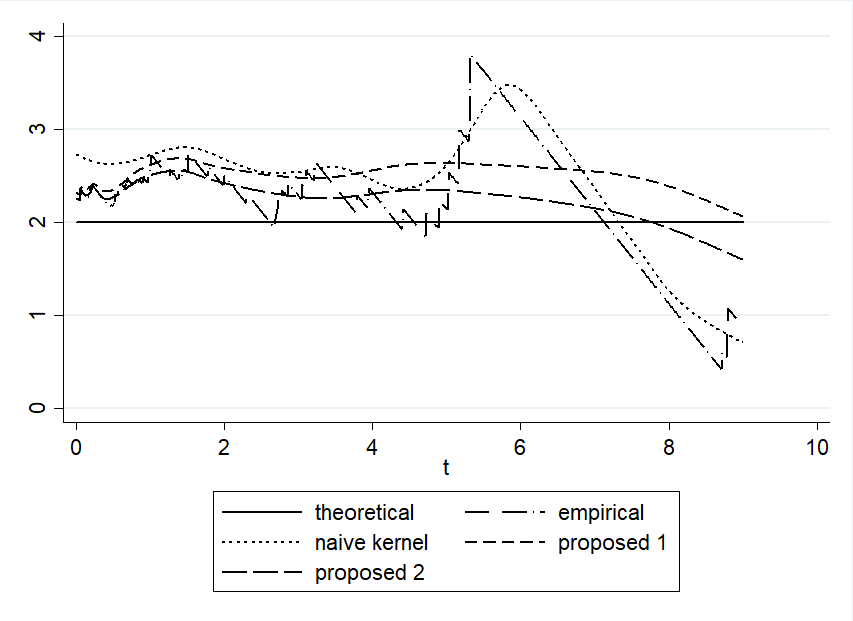}}}
	\subfigure[$X\sim Gamma(2,3)$\label{fig:mrlgam}]{
		\resizebox*{5.85cm}{!}{\includegraphics{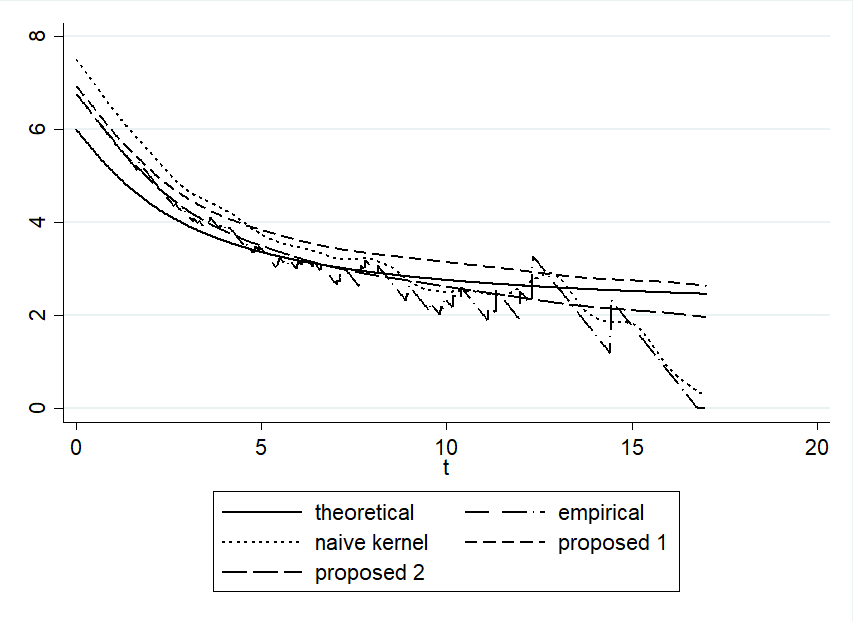}}}\\
	\subfigure[$X\sim Weibull(3,2)$\label{fig:mrlwei}]{
		\resizebox*{5.85cm}{!}{\includegraphics{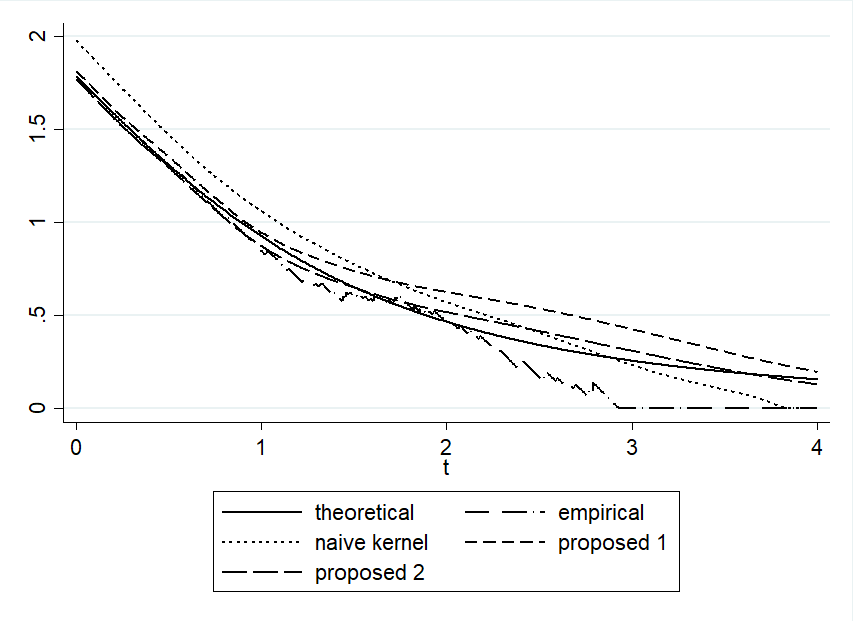}}}
	\subfigure[$X\sim abs.N(0,1)$\label{fig:mrlabs}]{
		\resizebox*{5.85cm}{!}{\includegraphics{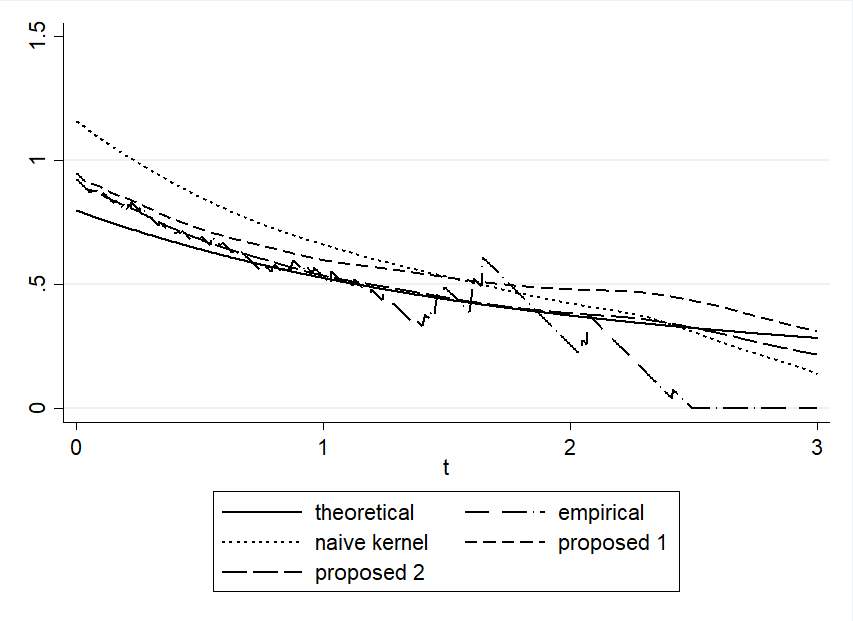}}}
	\caption{Graphs comparisons of $m_X(t)$, $m_n(t)$, $\widehat{m}_X(t)$, $\widetilde{m}_{X,1}(t)$, and $\widetilde{m}_{X,2}(t)$ for several distributions, with sample size $n=50$.}\label{fig:mrl}
\end{figure}

\begin{figure}
	\centering
	\subfigure[$X\sim exp\left(\frac{1}{2}\right)$\label{fig:biasexp}]{
		\resizebox*{5.85cm}{!}{\includegraphics{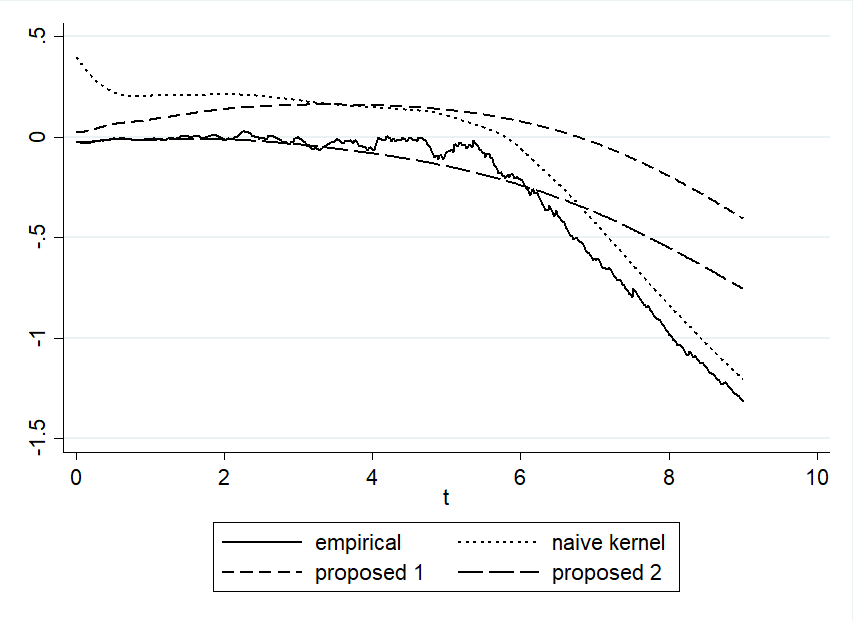}}}
	\subfigure[$X\sim abs.N(0,1)$\label{fig:biasabs}]{
		\resizebox*{5.85cm}{!}{\includegraphics{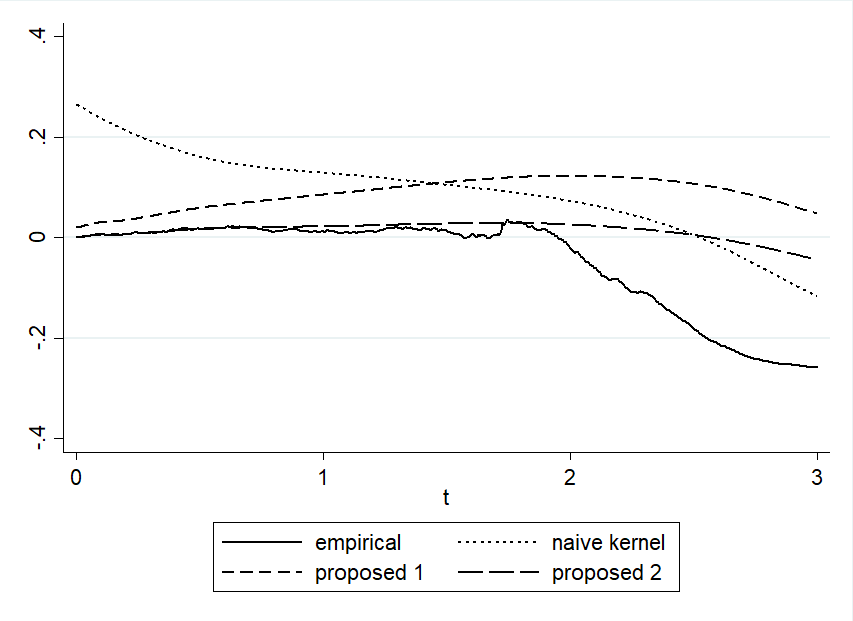}}}
	\caption{Simulated bias comparisons of $m_n(t)$, $\widehat{m}_X(t)$, $\widetilde{m}_{X,1}(t)$, and $\widetilde{m}_{X,2}(t)$ for several distributions, with sample size $n=50$ and $500$ repetitions.}\label{fig:bias}
\end{figure}

\subsection{real data analysis}
\label{sec:real}
In this analysis, we used the UIS Drug Treatment Study Data from \cite{HL98} to show the performances our proposed methods for real data. The data set records the result of an experiment about how long someone who got drug treatment to relapse (reuse) the drug again. The variable we used in the calculation is the "time" variable, which represents the number of days after the admission to drug treatment until drug relapse.

\begin{figure}
	\centering
	\includegraphics[width=0.75\textwidth]{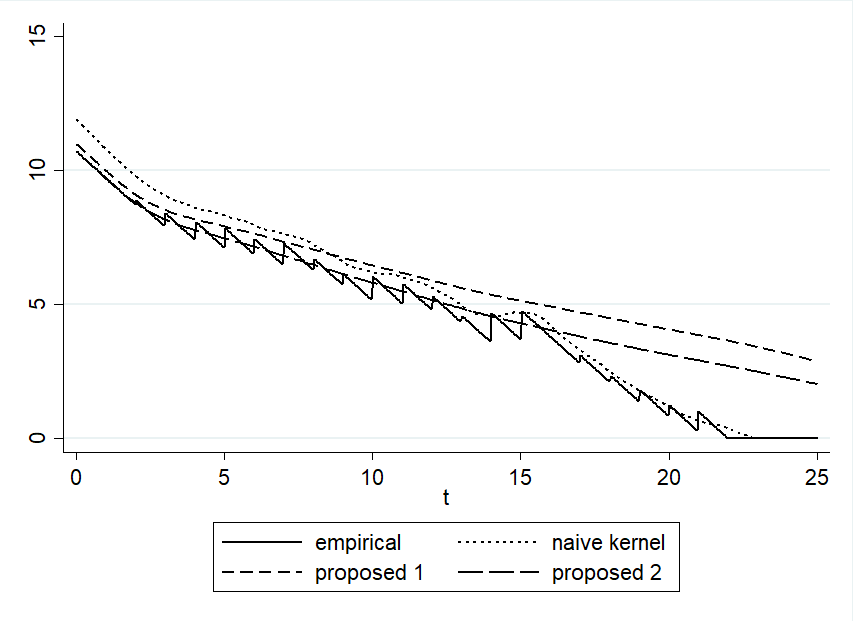}
	\caption{Comparison of $m_n(t)$, $\widehat{m}_X(t)$, $\widetilde{m}_{X,1}(t)$, and $\widetilde{m}_{X,2}(t)$ for UIS data}
	\label{fig:uis}
\end{figure}

(Figure~\ref{fig:uis}) shows that, once again, the naive kernel estimator is just a smoothed version of the empirical MRL function. Furthermore, soon after $m_n(t)$ touches $0$, $\widehat{m}_X(t)$ also reaches $0$. Conversely, though our proposed estimators are decreasing as well, but they are much slower than the other two.

\section{Conclusion}
\label{sec:conclude}
This article has proposed two new estimators for the mean residual life function (also the survival and the cumulative survival functions) when the data are supported not on the entire real line. First we constructed two new estimators for both the survival function and the cumulative survival function using bijective transformation. After deriving the formulas for the biases and the variances of the estimators, we defined two estimators for the MRL function. The properties of our proposed methods have been discovered and discussed. Moreover, the results of the numerical studies reveal the superior performance of the proposed estimators. For future research, establishing new estimators using a similar idea for other functions such as the distribution function, the hazard rate function, or regression of the MRL function will be a valuable contribution to this field.



\appendix

\section{Some lemmas needed to prove the theorems}
Though sometimes not stated explicitly in the proofs of our theorems, the following lemmas are needed for the calculations.
\begin{lemma}\label{lemma:integral}
	Under the condition C1, the following equations hold
	\begin{gather}
	\int_{-\infty}^\infty V(x)K(x)\mathrm{d}x=\frac{1}{2},\\
	\int_{-\infty}^\infty xV(x)K(x)\mathrm{d}x=-\frac{1}{2}\int_{-\infty}^\infty V(x)W(x)\mathrm{d}x,\\
	\int_{-\infty}^\infty \mathbb{V}(x)K(x)\mathrm{d}x=\int_{-\infty}^\infty V(x)W(x)\mathrm{d}x.
	\end{gather}
\end{lemma}
\begin{proof}
	All of the above equations can be proven using the integration by parts and the definitions of $V(x)$, $\mathbb{V}(x)$, and $W(x)$.
\end{proof}
\begin{lemma}\label{lemma:cv}
	Let $f_Y(t)$ and $S_Y(t)$ be the probability density function and the survival function of $Y=g^{-1}(X)$, and let $a(t)=\int_t^\infty g'(y)S_Y(y)\mathrm{d}y$ and $A(t)=\int_t^\infty g'(y)a(y)\mathrm{d}y$. Then, under the condition C6, we have for $t\in\Omega$,
	\begin{gather}
	f_Y(g^{-1}(t))=g'(g^{-1}(t))f_X(t),\label{eq:a4}\\
	S_Y(g^{-1}(t))=S_X(t),\label{eq:a5}\\
	a(g^{-1}(t))=\mathbb{S}_X(t),\label{eq:a6}\\
	A(g^{-1}(t))=\bar{\mathbb{S}}_X(t).\label{eq:a7}
	\end{gather}
\end{lemma}
\begin{proof}
	Using the change-of-variable technique, it is obvious that $f_Y(t)=g'(t)f_X(g(t))$, and it validifies equation (\ref{eq:a4}). For equation (\ref{eq:a5}), by the definition,
	\begin{eqnarray*}
		S_Y(g^{-1}(t))=\Pr[Y>g^{-1}(t)]=\Pr[g(Y)>t]=\Pr(X>t)=S_X(t).
	\end{eqnarray*}
	Equation (\ref{eq:a6}) is easily done using integration by $x=g(y)$ substitution, which is
	\begin{eqnarray*}
		a(g^{-1}(t))=\int_{g^{-1}(t)}^\infty g'(y)S_Y(y)\mathrm{d}y=\int_t^{\omega''}S_Y(g^{-1}(x))\mathrm{d}x=\int_t^{\omega''}S_X(x)\mathrm{d}x=\mathbb{S}_X(t).
	\end{eqnarray*}
	The same fashion goes for equation (\ref{eq:a7}).
\end{proof}
\begin{remark}\label{rem:cv}
	The ideas to construct our proposed estimators actually came from Lemma~\ref{lemma:cv}. We intentionally constructed the estimators of $S_X(t)$ and $\mathbb{S}_X(t)$ using the relationships stated at equation (\ref{eq:a5}) and equation (\ref{eq:a6}), respectively. We refer to this lemma as the change-of-variable properties.
\end{remark}
\begin{lemma}\label{lemma:a(t)}
	Let $a(t)=\int_t^\infty g'(y)S_Y(y)\mathrm{d}y$ and
	\begin{gather}
	\widehat{a}(t)=\int_t^\infty g'(y)\widehat{S}_Y(y)\mathrm{d}y=\frac{1}{n}\sum_{i=1}^n\int_t^\infty g'(y)V\left(\frac{y-Y_i}{h}\right)\mathrm{d}y
	\end{gather}
	be the naive kernel estimator of $a(t)$. If $B\subset\mathbb{R}$ is an interval where both $\widehat{a}(t)$ and $a(t)$ are bounded, then $\sup_{t\in B}|\widehat{a}(t)-a(t)|\rightarrow_{a.s.}0$.
\end{lemma}
\begin{proof}
	Since $\widehat{a}(t)$ and $a(t)$ are both bounded, non-increasing, and continuous on $B$, then for any $\varepsilon>0$, we can find $k$ number of points on $A$ such that
	\begin{gather*}
	-\infty\leq\inf A=t_1<t_2<...<t_k=\sup A\leq\infty,
	\end{gather*}
	and $a(t_j)-a(t_{j+1})\leq\varepsilon/2$, $j=1,2,...,k-1$. For any $t\in B$, it is clear that there exists $j$ such that $t_j\leq t<t_{j+1}$. For that particular $j$, we have
	\begin{gather*}
	\widehat{a}(t_j)\geq\widehat{a}(t)\geq\widehat{a}(t_{j+1}) \; \; \; \mathrm{and} \; \; \; a(t_j)\geq a(t)\geq a(t_{j+1}),
	\end{gather*}
	which result in
	\begin{gather*}
	\widehat{a}(t_{j+1})-a(t_{j+1})-\frac{\varepsilon}{2}\leq\widehat{a}(t)-a(t)\leq\widehat{a}(t_j)-a(t_j)+\frac{\varepsilon}{2}.
	\end{gather*}
	Therefore,
	\begin{gather*}
	\sup_{t\in B}|\widehat{a}(t)-a(t)|\leq\sup_j|\widehat{a}(t_j)-a(t_j)|+\varepsilon.
	\end{gather*}
	Now, because $\widehat{a}(t)$ is a naive kernel estimator, it is clear that for fix $t_0$, $\widehat{a}(t_0)$ converges almost surely to $a(t_0)$. Thus, we get $|\widehat{a}(t_0)-a(t_0)|\rightarrow_{a.s.}0$. Hence, for any $\varepsilon>0$, almost surely $\sup_{t\in B}|\widehat{a}(t)-a(t)|\leq\varepsilon$ when $n\rightarrow\infty$, which concludes the proof.
\end{proof}

\section{Proof of theorem~\ref{thm:set1}}
Utilizing the usual reasoning of \textit{i.i.d.} random variables and the transformation property of expectation, and with the fact
\begin{gather*}
V_{1,h}(x,y)=\frac{1}{h}\int_{-\infty}^y K\left(\frac{x-z}{h}\right)\mathrm{d}z,
\end{gather*}
we have
\begin{eqnarray*}
	E[\widetilde{S}_{X,1}(t)]&=&E[V_{1,h}(g^{-1}(t),g^{-1}(X_1))]\\
	&=&\int_{-\infty}^\infty V_{1,h}(g^{-1}(t),y)f_Y(y)\mathrm{d}y\\
	&=&\frac{1}{h}\int_{-\infty}^\infty K\left(\frac{g^{-1}(t)-y}{h}\right)S_Y(y)\mathrm{d}y\\
	&=&\int_{-\infty}^\infty S_Y(g^{-1}(t)-hu)K(u)\mathrm{d}u\\
	&=&\int_{-\infty}^\infty\left[S_Y(g^{-1}(t))+hu f_Y(g^{-1}(t))-\frac{h^2}{2}u^2 f_Y'(g^{-1}(t))+o(h^2)\right]K(u)\mathrm{d}u\\
	&=&S_X(t)-\frac{h^2}{2}b_1(t)\int_{-\infty}^\infty u^2 K(u)\mathrm{d}u+o(h^2),
\end{eqnarray*}
and we have $Bias[\widetilde{S}_{X,1}(t)]$. For the variance of $\widetilde{S}_{X,1}(t)$, we first calculate
\begin{eqnarray*}
	E[V_{1,h}^2(g^{-1}(t),g^{-1}(X_1))]&=&\int_{-\infty}^\infty V_{1,h}^2(g^{-1}(t),y)f_Y(y)\mathrm{d}y\\
	&=&\frac{2}{h}\int_{-\infty}^\infty V_{1,h}(g^{-1}(t),y)K\left(\frac{g^{-1}(t)-y}{h}\right)S_Y(y)\mathrm{d}y\\
	&=&2\int_{-\infty}^\infty[S_Y(g^{-1}(t))+huf_Y(g^{-1}(t))+o(h)]V(u)K(u)\mathrm{d}u\\
	&=&S_X(t)-hg'(g^{-1}(t))f_X(t)\int_{-\infty}^\infty V(u)W(u)\mathrm{d}u+o(h).
\end{eqnarray*}
Hence, the variance is
\begin{eqnarray*}
	Var[\widetilde{S}_{X,1}(t)]&=&\frac{1}{n}[E\{V_{1,h}^2(g^{-1}(t),g^{-1}(X_1))\}-E^2\{V_{1,h}(g^{-1}(t),g^{-1}(t))\}]\\
	&=&\frac{1}{n}S_X(t)F_X(t)-\frac{h}{n}g'(g^{-1}(t))f_X(t)\int_{-\infty}^\infty V(y)W(y)\mathrm{d}y+o\left(\frac{h}{n}\right).
\end{eqnarray*}

For the calculation of $Bias[\widetilde{\mathbb{S}}_{X,1}(t)]$, recall that
\begin{gather*}
\mathbb{V}_{1,h}(g^{-1}(t),y)=\int_t^{\omega''}V_{1,h}(g^{-1}(z),y)\mathrm{d}z,
\end{gather*}
and by assuming we can change the order of the integral signs, we get
\begin{eqnarray*}
	E[\widetilde{\mathbb{S}}_{X,1}(t)]&=&\int_{-\infty}^\infty\int_t^{\omega''}V_{1,h}(g^{-1}(z),y)\mathrm{d}zf_Y(y)\mathrm{d}y\\
	&=&\int_t^{\omega''}E[V_{1,h}(g^{-1}(z),Y)]\mathrm{d}z\\
	&=&\mathbb{S}_X(t)-\frac{h^2}{2}\int_t^{\omega''}b_1(z)\mathrm{d}z\int_{-\infty}^\infty y^2 K(y)\mathrm{d}y+o(h^2).
\end{eqnarray*}
It is easy to see $b_2(t)=-\int_t^{\omega''}b_1(z)\mathrm{d}z$, and then the formula of $Bias[\widetilde{\mathbb{S}}_{X,1}(t)]$ is done.

Before calculating $Var[\widetilde{\mathbb{S}}_{X,1}(t)]$, we must first note that
\begin{eqnarray*}
	\frac{\mathrm{d}}{\mathrm{d}y}\mathbb{V}_{1,h}(x,y)&=&\frac{1}{h}\int_x^\infty g'(z)K\left(\frac{z-y}{h}\right)\mathrm{d}z\\
	&=&\int_{\frac{x-y}{h}}^\infty g'(y+hz)K(z)\mathrm{d}z\\
	&=&g'(y)V\left(\frac{x-y}{h}\right)+hg''(y)\int_{\frac{x-y}{h}}^\infty zK(z)\mathrm{d}z+...\\
	&=&g'(y)V\left(\frac{x-y}{h}\right)+o(h)
\end{eqnarray*}
and
\begin{eqnarray*}
	\mathbb{V}_{1,h}(x,y)&=&h\int_{\frac{x-y}{h}}^\infty g'(y+hz)V(z)\mathrm{d}z\\
	&=&hg'(y)\mathbb{V}\left(\frac{x-y}{h}\right)+h^2 g''(y)\int_{\frac{x-y}{h}}^\infty zV(z)\mathrm{d}z+...\\
	&=&hg'(y)\mathbb{V}\left(\frac{x-y}{h}\right)+o(h).
\end{eqnarray*}
Now, we get
\begin{eqnarray*}
	E[\mathbb{V}_{1,h}^2(g^{-1}(t),g^{-1}(X_1))]&=&2\int_{-\infty}^\infty\mathbb{V}_{1,h}(g^{-1}(t),y)V\left(\frac{g^{-1}(t)-y}{h}\right)g'(y)S_Y(y)\mathrm{d}y+o(h)\\
	&=&2\int_{-\infty}^\infty\bigg[g'(y)V^2\left(\frac{g^{-1}(t)-y}{h}\right)\\
	&& \; \; \; \; \; \; \; +\frac{1}{h}\mathbb{V}_{1,h}(g^{-1}(t),y)K\left(\frac{g^{-1}(t)-y}{h}\right)\bigg]a(y)\mathrm{d}y+o(h).
\end{eqnarray*}
Conducting integration by parts once again for the first term, we have
\begin{eqnarray*}
	&&2\int_{-\infty}^\infty g'(y)V^2\left(\frac{g^{-1}(t)-y}{h}\right)a(y)\mathrm{d}y\\
	&& \; \; \; \; \; \; \; \; \; \; \; \; =\frac{4}{h}\int_{-\infty}^\infty V\left(\frac{g^{-1}(t)-y}{h}\right)K\left(\frac{g^{-1}(t)-y}{h}\right)A(y)\mathrm{d}y\\
	&& \; \; \; \; \; \; \; \; \; \; \; \; =4\int_{-\infty}^\infty A(g^{-1}(t)-hu)V(u)K(u)\mathrm{d}u\\
	&& \; \; \; \; \; \; \; \; \; \; \; \; =2\bar{\mathbb{S}}_X(t)-2hg'(g^{-1}(t))\mathbb{S}_X(t)\int_{-\infty}^\infty V(u)W(u)\mathrm{d}u+o(h).
\end{eqnarray*}
And the second term can be calculated with
\begin{eqnarray*}
	&&\frac{2}{h}\int_{-\infty}^\infty\mathbb{V}_{1,h}(g^{-1}(t),y)K\left(\frac{g^{-1}(t)-y}{h}\right)a(y)\mathrm{d}y\\
	&& \; \; \; \; \; \; \; \; \; \; \; \; =2\int_{-\infty}^\infty\left[g'(y)\mathbb{V}\left(\frac{g^{-1}(t)-y}{h}\right)+o(1)\right]K\left(\frac{g^{-1}(t)-y}{h}\right)a(y)\mathrm{d}y\\
	&& \; \; \; \; \; \; \; \; \; \; \; \; =2h\int_{-\infty}^\infty g'(g^{-1}(t)-hu)a(g^{-1}(t)-hu)\mathbb{V}(y)K(y)\mathrm{d}y+o(h)\\
	&& \; \; \; \; \; \; \; \; \; \; \; \; =2h\int_{-\infty}^\infty [g'(g^{-1}(t))+o(1)][a(g^{-1}(t))+o(1)]\mathbb{V}(y)K(y)\mathrm{d}y+o(h)\\
	&& \; \; \; \; \; \; \; \; \; \; \; \; =2hg'(g^{-1}(t))\mathbb{S}_X(t)\int_{-\infty}^\infty V(u)W(u)\mathrm{d}u+o(h).
\end{eqnarray*}
Thus, we get
\begin{gather*}
E[\mathbb{V}_{1,h}^2(g^{-1}(t),g^{-1}(X_1))]=2\bar{\mathbb{S}}_X(t)+o(h),
\end{gather*}
and this easily proves the $Var[\widetilde{\mathbb{S}}_{X,1}(t)]$ formula.

Before going into the calculation of the covariance, we have to take a look at
\begin{eqnarray*}
	&&E[\mathbb{V}_{1,h}(g^{-1}(t),g^{-1}(X_1))V_{1,h}(g^{-1}(t),g^{-1}(X_1))]\\
	&& \; \; \; \; \; \; \; \; \; \; \; \; \; =\int_{-\infty}^\infty \bigg[g'(y)V\left(\frac{g^{-1}(t)-y}{h}\right)V_{1,h}(g^{-1}(t),y)\\
	&& \; \; \; \; \; \; \; \; \; \; \; \; \; \; \; \; \; \; \; \; +\frac{1}{h}\mathbb{V}_{1,h}(g^{-1}(t),y)K\left(\frac{g^{-1}(t)-y}{h}\right)\bigg]S_Y(y)\mathrm{d}y.
\end{eqnarray*}
Once again we need to calculate them separately. The first term is
\begin{eqnarray*}
	&&\int_{-\infty}^\infty g'(y)V\left(\frac{g^{-1}(t)-y}{h}\right)V_{1,h}(g^{-1}(t),y)S_Y(y)\mathrm{d}y\\
	&& \; \; \; \; \; \; \; \; \; \; \; \; \; =\int_{-\infty}^\infty\bigg[\frac{1}{h^2}K\left(\frac{g^{-1}(t)-y}{h}\right)\int_{-\infty}^y K\left(\frac{g^{-1}(t)-z}{h}\right)\mathrm{d}z\\
	&& \; \; \; \; \; \; \; \; \; \; \; \; \; \; \; \; \; \; \; \; +\frac{1}{h}V\left(\frac{g^{-1}(t)-y}{h}\right)K\left(\frac{g^{-1}(t)-y}{h}\right)\bigg]a(y)\mathrm{d}y\\
	&& \; \; \; \; \; \; \; \; \; \; \; \; \; =\int_{-\infty}^\infty a(g^{-1}(t)-hu)\left[K(u)\int_u^\infty K(v)\mathrm{d}v+V(u)K(u)\right]\mathrm{d}u\\
	&& \; \; \; \; \; \; \; \; \; \; \; \; \; =\mathbb{S}_X(t)-hg'g^{-1}(t))S_X(t)\int_{-\infty}^\infty V(u)W(u)\mathrm{d}u+o(h),
\end{eqnarray*}
while the second term is
\begin{eqnarray*}
	&&\frac{1}{h}\int_{-\infty}^\infty\mathbb{V}_{1,h}(g^{-1}(t),y)K\left(\frac{g^{-1}(t)-y}{h}\right)S_Y(y)\mathrm{d}y\\
	&& \; \; \; \; \; \; \; \; \; \; \; \; \; =h\int_{-\infty}^\infty g'(g^{-1}(t)-hu)S_Y(g^{-1}(t)-hu)\mathbb{V}(u)K(u)\mathrm{d}u+o(h)\\
	&& \; \; \; \; \; \; \; \; \; \; \; \; \; =h\int_{-\infty}^\infty[g'(g^{-1}(t))+o(1)][S_Y(g^{-1}(t))+o(1)]\mathbb{V}(u)K(u)\mathrm{d}u+o(h)\\
	&& \; \; \; \; \; \; \; \; \; \; \; \; \; =hg'(g^{-1}(t))S_X(t)\int_{-\infty}^\infty V(u)W(u)\mathrm{d}u+o(h),
\end{eqnarray*}
then we have
\begin{eqnarray*}
	E[\mathbb{V}_{1,h}(g^{-1}(t),g^{-1}(X_1))V_{1,h}(g^{-1}(t),g^{-1}(X_1))]=\mathbb{S}_X(t)+o(h).
\end{eqnarray*}
Hence, the covariance is
\begin{eqnarray*}
	Cov[\widetilde{\mathbb{S}}_{X,1}(t),\widetilde{S}_{X,1}(t)]&=&\frac{1}{n}[E\{\mathbb{V}_{1,h}(g^{-1}(t),g^{-1}(X_1))V_{1,h}(g^{-1}(t),g^{-1}(X_1))\}\\
	&& \; \; \; \; \; \; \; -E\{\mathbb{V}_{1,h}(g^{-1}(t),g^{-1}(X_1))\}\{V_{1,h}(g^{-1}(t),g^{-1}(X_1))\}]\\
	&=&\frac{1}{n}[\mathbb{S}_X(t)-\mathbb{S}_X(t)S_X(t)+o(h)]\\
	&=&\frac{1}{n}\mathbb{S}_X(t)F_X(t)+o\left(\frac{h}{n}\right).
\end{eqnarray*}

\section{Proof of theorem~\ref{thm:set2}}
The usual reasoning of \textit{i.i.d.} random variables and the transformation property of expectation result in
\begin{eqnarray*}
	E[\widetilde{S}_{X,2}(t)]&=&\int_{-\infty}^\infty V\left(\frac{g^{-1}(t)-y}{h}\right)f_Y(y)\mathrm{d}y\\
	&=&\int_{-\infty}^\infty K\left(\frac{g^{-1}(t)-y}{h}\right)S_Y(y)\mathrm{d}y\\
	&=&\int_{-\infty}^\infty S_Y(g^{-1}(t)-hu)K(u)\mathrm{d}u\\
	&=&S_X(t)-\frac{h^2}{2}b_1(t)\int_{-\infty}^\infty u^2 K(u)\mathrm{d}u+o(h^2),
\end{eqnarray*}
and this gives us the $Bias[\widetilde{S}_{X,2}(t)]$. For the variance of $\widetilde{S}_{X,2}(t)$, first we calculate
\begin{eqnarray*}
	E[V_{2,h}^2(g^{-1}(t),g^{-1}(X_1))]&=&\frac{2}{h}\int_{-\infty}^\infty V\left(\frac{g^{-1}(t)-y}{h}\right)K\left(\frac{g^{-1}(t)-y}{h}\right)S_Y(y)\mathrm{d}y\\
	&=&2\int_{-\infty}^\infty S_Y(g^{-1}(t)-hu)V(u)K(u)\mathrm{d}u\\
	&=&S_X(t)-hg'(g^{-1}(t))f_X(t)\int_{-\infty}^\infty V(u)W(u)\mathrm{d}u+o(h).
\end{eqnarray*}
The resulting variance is
\begin{eqnarray*}
	Var[\widetilde{S}_{X,2}(t)]=\frac{1}{n}S_X(t)F_X(t)-\frac{h}{n}g'(g^{-1}(t))f_X(t)\int_{-\infty}^\infty V(y)W(y)\mathrm{d}y+o\left(\frac{h}{n}\right).
\end{eqnarray*}

Next for $Bias[\widetilde{\mathbb{S}}_{X,2}(t)]$, utilizing similar reasoning as before, we get
\begin{eqnarray*}
	E[\widetilde{\mathbb{S}}_{X,2}(t)]&=&\int_{-\infty}^\infty\mathbb{V}_{2,h}(g^{-1}(t),y)f_Y(y)\mathrm{d}y\\
	&=&\int_{-\infty}^\infty g'(y)V\left(\frac{g^{-1}(t)-y}{h}\right)S_Y(y)\mathrm{d}y\\
	&=&\frac{1}{h}\int_{-\infty}^\infty K\left(\frac{g^{-1}(t)-y}{h}\right)a(y)\mathrm{d}y\\
	&=&\int_{-\infty}^\infty a(g^{-1}(t)-hu)K(u)\mathrm{d}u\\
	&=&\mathbb{S}_X(t)+\frac{h^2}{2}b_3(t)\int_{-\infty}^\infty u^2 K(u)\mathrm{d}u+o(h^2),
\end{eqnarray*}
and this proves the bias part. For the variance, we need
\begin{eqnarray*}
	E[\mathbb{V}_{2,h}^2(g^{-1}(t),g^{-1}(X_1))]&=&2\int_{-\infty}^\infty g'(y)\mathbb{V}_{2,h}(g^{-1}(t),y)V\left(\frac{g^{-1}(t)-y}{h}\right)S_Y(y)\mathrm{d}y\\
	&=&2\int_{-\infty}^\infty\bigg[g'(y)V^2\left(\frac{g^{-1}(t)-y}{h}\right)\\
	&& \; \; \; \; \; \; \; +\frac{1}{h}\mathbb{V}_{2,h}(g^{-1}(t),y)K\left(\frac{g^{-1}(t)-y}{h}\right)\bigg]a(y)\mathrm{d}y.
\end{eqnarray*}
Once again using the integration by parts for the first term, we have
\begin{eqnarray*}
	2\int_{-\infty}^\infty g'(y)V^2\left(\frac{g^{-1}(t)-y}{h}\right)a(y)\mathrm{d}y&=&\frac{4}{h}\int_{-\infty}^\infty V\left(\frac{g^{-1}(t)-y}{h}\right)K\left(\frac{g^{-1}(t)-y}{h}\right)A(y)\mathrm{d}y\\
	&=&4\int_{-\infty}^\infty A(g^{-1}(t)-hu)V(u)K(u)\mathrm{d}u\\
	&=&2\bar{\mathbb{S}}_X(t)-2hg'(g^{-1}(t))\mathbb{S}_X(t)\int_{-\infty}^\infty V(u)W(u)\mathrm{d}u+o(h).
\end{eqnarray*}
The second term can be calculated in a similar way, which is
\begin{eqnarray*}
	&&\frac{2}{h}\int_{-\infty}^\infty\mathbb{V}_{2,h}(g^{-1}(t),y)K\left(\frac{g^{-1}(t)-y}{h}\right)a(y)\mathrm{d}y\\
	&& \; \; \; \; \; \; \; \; \; \; \; \; \; =2h\int_{-\infty}^\infty a(g^{-1}(t)-hu)\int_u^\infty g'(g^{-1}(t)-hv)V(v)\mathrm{d}vK(u)\mathrm{d}u\\
	&& \; \; \; \; \; \; \; \; \; \; \; \; \; =2h\int_{-\infty}^\infty [a(g^{-1}(t))+o(1)]\int_u^\infty[g'(g^{-1}(t))+o(1)]V(v)\mathrm{d}vK(u)\mathrm{d}u\\
	&& \; \; \; \; \; \; \; \; \; \; \; \; \; =2hg'(g^{-1}(t))\mathbb{S}_X(t)\int_{-\infty}^\infty V(u)W(u)\mathrm{d}u+o(h).
\end{eqnarray*}
Hence, we get
\begin{gather*}
E[\mathbb{V}_{2,h}^2(g^{-1}(t),g^{-1}(X_1))]=2\bar{\mathbb{S}}_X(t)+o(h),
\end{gather*}
proving the formula of $Var[\widetilde{\mathbb{S}}_{X,2}(t)]$.

Before moving onto the calculation of the covariance, we have to take a look at
\begin{eqnarray*}
	&&E[\mathbb{V}_{2,h}(g^{-1}(t),g^{-1}(X_1))V_{2,h}(g^{-1}(t),g^{-1}(X_1))]\\
	&& \; \; \; \; \; \; \; \; \; \; \; \; \; =\int_{-\infty}^\infty \left[g'(y)V^2\left(\frac{g^{-1}(t)-y}{h}\right)+\frac{1}{h}\mathbb{V}_{2,h}(g^{-1}(t),y)K\left(\frac{g^{-1}(t)-y}{h}\right)\right]S_Y(y)\mathrm{d}y.
\end{eqnarray*}
Once again we need to calculate them separately. The first term is
\begin{eqnarray*}
	&&\int_{-\infty}^\infty g'(y)V^2\left(\frac{g^{-1}(t)-y}{h}\right)S_Y(y)\mathrm{d}y\\
	&& \; \; \; \; \; \; \; \; \; \; \; \; \; =\frac{2}{h}\int_{-\infty}^\infty V\left(\frac{g^{-1}(t)-y}{h}\right)K\left(\frac{g^{-1}(t)-y}{h}\right)a(y)\mathrm{d}y\\
	&& \; \; \; \; \; \; \; \; \; \; \; \; \; =2\int_{-\infty}^\infty a(g^{-1}(t)-hu)V(y)K(y)\mathrm{d}u\\
	&& \; \; \; \; \; \; \; \; \; \; \; \; \; =\mathbb{S}_X(t)-hg'(g^{-1}(t))S_X(t)\int_{-\infty}^\infty V(u)W(u)\mathrm{d}u+o(h),
\end{eqnarray*}
while the second term is
\begin{eqnarray*}
	&&\frac{1}{h}\int_{-\infty}^\infty\mathbb{V}_{2,h}(g^{-1}(t),y)K\left(\frac{g^{-1}(t)-y}{h}\right)S_Y(y)\mathrm{d}y\\
	&& \; \; \; \; \; \; \; \; \; \; \; \; \; =h\int_{-\infty}^\infty S_Y(g^{-1}(t)-hu)\int_u^\infty g'(g^{-1}(t)-hv)V(v)\mathrm{d}vK(u)\mathrm{d}u\\
	&& \; \; \; \; \; \; \; \; \; \; \; \; \; =h\int_{-\infty}^\infty[S_Y(g^{-1}(t))+o(1)][g'(g^{-1}(t))\mathbb{V}(u)+o(1)]K(u)\mathrm{d}u\\
	&& \; \; \; \; \; \; \; \; \; \; \; \; \; =hg'(g^{-1}(t))S_X(t)\int_{-\infty}^\infty V(u)W(u)\mathrm{d}u+o(h),
\end{eqnarray*}
and the result is
\begin{eqnarray*}
	E[\mathbb{V}_{2,h}(g^{-1}(t),g^{-1}(X_1))V_{2,h}(g^{-1}(t),g^{-1}(X_1))]=\mathbb{S}_X(t)+o(h).
\end{eqnarray*}
Hence, the covariance is
\begin{eqnarray*}
	Cov[\widetilde{\mathbb{S}}_{X,2}(t),\widetilde{S}_{X,2}(t)]=\frac{1}{n}\mathbb{S}_X(t)F_X(t)+o\left(\frac{h}{n}\right).
\end{eqnarray*}

\section{Proof of theorem~\ref{thm:bmrlf}}
As for a fixed $t$ we have that $\widetilde{S}_{X,1}(t)$ and $\widetilde{\mathbb{S}}_{X,1}(t)$ are consistent estimators for $S_X(t)$ and $\mathbb{S}_X(t)$, respectively, then
\begin{eqnarray*}
	\widetilde{m}_{X,1}(t)-m_X(t)&=&\frac{\widetilde{\mathbb{S}}_{X,1}(t)-\widetilde{S}_{X,1}(t)m_X(t)}{S_X(t)}\left[1+\frac{S_X(t)-\widetilde{S}_{X,1}(t)}{\widetilde{S}_{X,1}(t)}\right]\\
	&=&\frac{\widetilde{\mathbb{S}}_{X,1}(t)-\widetilde{S}_{X,1}(t)m_X(t)}{S_X(t)}[1+o_p(1)].
\end{eqnarray*}
Thus, using Theorem~\ref{thm:set1}, we get
\begin{eqnarray*}
	Bias[\widetilde{m}_{X,1}(t)]&=&\frac{1}{S_X(t)}[E\{\widetilde{\mathbb{S}}_{X,1}(t)\}-m_X(t)E\{\widetilde{S}_{X,1}(t)\}]\\
	&=&\frac{h^2}{2S_X(t)}[b_2(t)+m_X(t)b_1(t)]\int_{-\infty}^\infty y^2 K(y)\mathrm{d}y+o(h^2).
\end{eqnarray*}
The same argument easily proves the formula of $Bias[\widetilde{m}_{X,2}(t)]$.

Using a similar method, for $i=1,2$, we have
\begin{eqnarray*}
	Var[\widetilde{m}_{X,i}(t)]&=&Var[\widetilde{m}_{X,i}(t)-m_X(t)]\\
	&=& Var\left[\frac{\widetilde{\mathbb{S}}_{X,1}(t)-\widetilde{S}_{X,1}(t)m_X(t)}{S_X(t)}\right]\\
	&=&\frac{1}{S_X^2(t)}[Var\{\widetilde{\mathbb{S}}_{X,i}(t)\}+m_X^2(t)Var\{\widetilde{S}_{X,i}(t)\}-2m_X(t)Cov\{\widetilde{\mathbb{S}}_{X,i},\widetilde{S}_{X,i}(t)\}]\\
	&=&\frac{1}{n}\frac{b_4(t)}{S_X^2(t)}-\frac{h}{n}\frac{b_5(t)}{S_X^2(t)}\int_{-\infty}^\infty V(y)W(y)\mathrm{d}y+o\left(\frac{h}{n}\right).
\end{eqnarray*}

\section{Proof of theorem~\ref{thm:normal}}
Because the proof of the case $i=1$ is similar, we will only explain the case of $i=2$ in detail. First, for some $\delta>0$, using H\"{o}lder and $c_r$ inequalities, we have
\begin{gather*}
E[|V_{2,h}(g^{-1}(t),g^{-1}(X_1))-E\{V_{2,h}(g^{-1}(t),g^{-1}(X_1))\}|^{2+\delta}]\leq 2^{2+\delta}E[|V_{2,h}(g^{-1}(t),g^{-1}(X_1))|^{2+\delta}].
\end{gather*}
But, since $0\leq V_{2,h}(x,y)\leq 1$ for any $x,y\in\mathbb{R}$, then
\begin{gather*}
E[|V_{2,h}(g^{-1}(t),g^{-1}(X_1))-E\{V_{2,h}(g^{-1}(t),g^{-1}(X_1))\}|^{2+\delta}]\leq 2^{2+\delta}<\infty,
\end{gather*}
and because $Var[V_{2,h}(g^{-1}(t),g^{-1}(X_1))]=O(1)$, we get
\begin{gather*}
\frac{E[|V_{2,h}(g^{-1}(t),g^{-1}(X_1))-E\{V_{2,h}(g^{-1}(t),g^{-1}(X_1))\}|^{2+\delta}]}{n^{\delta/2}[Var\{V_{2,h}(g^{-1}(t),g^{-1}(X_1))\}]^{1+\delta/2}}\rightarrow 0
\end{gather*}
when $n\rightarrow\infty$. Hence, by Loeve~\cite{loe63}, and with the fact $\widetilde{S}_{X,2}(t)\rightarrow_p S_X(t)$, we can conclude that
\begin{gather*}
\frac{\widetilde{S}_{X,2}(t)-S_X(t)}{\sqrt{Var[\widetilde{S}_{X,2}(t)]}}\rightarrow_D N(0,1).
\end{gather*}
Next, with a similar reasoning as before, we have
\begin{gather*}
E[|\mathbb{V}_{2,h}(g^{-1}(t),g^{-1}(X_1))-E\{\mathbb{V}_{2,h}(g^{-1}(t),g^{-1}(X_1))\}|^{2+\delta}]\leq 2^{2+\delta}E[|\mathbb{V}_{2,h}(g^{-1}(t),g^{-1}(X_1))|^{2+\delta}],
\end{gather*}
which, by the same inequalities, results in
\begin{eqnarray*}
	&&E[|\mathbb{V}_{2,h}(g^{-1}(t),g^{-1}(X_1))-E\{\mathbb{V}_{2,h}(g^{-1}(t),g^{-1}(X_1))\}|^{2+\delta}]\\
	&& \; \; \; \; \; \; \; \; \; \; \leq 2^{2+\delta}E\left[\left|\int_{-\infty}^{g^{-1}(X_1)}g'(z)V\left(\frac{g^{-1}(t)-z}{h}\right)\mathrm{d}z\right|^{2+\delta}\right]\\
	&& \; \; \; \; \; \; \; \; \; \; \leq 2^{2+\delta}E\left[\left|\int_{-\infty}^{g^{-1}(X_1)}g'(z)\mathrm{d}z\right|^{2+\delta}\right]\\
	&& \; \; \; \; \; \; \; \; \; \; \leq 2^{2+\delta}E(X_1^{2+\delta})\\
	&& \; \; \; \; \; \; \; \; \; \; <\infty.
\end{eqnarray*}
Therefore, with the same argument, we get
\begin{gather*}
\frac{\widetilde{\mathbb{S}}_{X,2}(t)-\mathbb{S}_X(t)}{\sqrt{Var[\widetilde{\mathbb{S}}_{X,2}(t)]}}\rightarrow_D N(0,1).
\end{gather*}
At last, by Slutsky's Theorem for rational function, the theorem is proven.
\begin{remark}\label{rem:delta}
	Since we only assume the existence of $E(X^3)$, then we should choose $\delta\leq 1$ in this proof.
\end{remark}

\section{Proof of theorem~\ref{thm:consistent}}
Nadaraya~\cite{nad64} guarantees that $\sup_{t\in\mathbb{R}}|\widehat{S}_Y(t)-S_Y(t)|\rightarrow_{a.s.}0$, which implies
\begin{gather*}
\sup_{t\in\Omega}|\widehat{S}_Y(g^{-1}(t))-S_Y(g^{-1}(t))|\rightarrow_{a.s.}0.
\end{gather*}
However, because $S_Y(g^{-1}(t))=S_X(t)$, and it is clear that $\widehat{S}_Y(g^{-1}(t))=\widetilde{S}_{X,1}(t)$,
then $\sup_{t\in\Omega}|\widetilde{S}_{X,1}(t)-S_X(t)|\rightarrow_{a.s.}0$ holds.

Next, since $\mathbb{S}_X(t)\geq 0$ is bounded above with
\begin{gather*}
\sup_{t\in\Omega}\mathbb{S}_X(t)=\lim_{t\rightarrow\omega'^+}\mathbb{S}_X(t)=E(X)-\omega',
\end{gather*}
then $a(g^{-1}(t))=\mathbb{S}_X(t)$ is bounded on $\Omega$. Furthermore,
\begin{eqnarray*}
	\widehat{a}(g^{-1}(t))&=&\frac{1}{n}\sum_{i=1}^n\int_{g^{-1}(t)}^\infty g'(z)V\left(\frac{z-g^{-1}(X_i)}{h}\right)\mathrm{d}z\\
	&=&\frac{1}{n}\sum_{i=1}^n\int_t^{\omega''} V\left(\frac{g^{-1}(z)-g^{-1}(X_i)}{h}\right)\mathrm{d}z\\
	&=&\widetilde{\mathbb{S}}_{X,1}(t)\\
	&>&0
\end{eqnarray*}
is also bounded above almost surely with
\begin{gather*}
\sup_{t\in\Omega}\widetilde{\mathbb{S}}_{X,1}(t)=\lim_{t\rightarrow\omega'^+}\widetilde{\mathbb{S}}_{X,1}(t)=\bar{X}-\omega'+O_p(h^2).
\end{gather*}
Thus, Lemma~\ref{lemma:a(t)} implies $\sup_{t\in\Omega}|\widehat{a}(g^{-1}(t))-a(g^{-1}(t))|\rightarrow_{a.s.} 0$, which is equivalent to $\sup_{t\in\Omega}|\widetilde{\mathbb{S}}_{X,1}(t)-\mathbb{S}_X(t)|\rightarrow_{a.s.}0$. As a conclusion, $\sup_{t\in\Omega}|\widetilde{m}_{X,1}(t)-m_X(t)|\rightarrow_{a.s.}0$ holds. The proof for the case of $i=2$ is similar.

\section{Proof of theorem~\ref{thm:preserve}}
Because, for $i=1,2$,
\begin{gather*}
\lim_{t\rightarrow\omega'^+}\widetilde{m}_{X,i}(t)=\frac{\lim_{t\rightarrow\omega'^+}\widetilde{\mathbb{S}}_{X,i}(t)}{\lim_{t\rightarrow\omega'^+}\widetilde{S}_{X,i}(t)},
\end{gather*}
we only need to see the limit behaviour of each estimators of the survival function and the cumulative survival function. First, we have
\begin{eqnarray*}
	\lim_{t\rightarrow\omega'^+}\widetilde{S}_{X,1}(t)&=&\frac{1}{nh}\sum_{i=1}^n\lim_{t\rightarrow\omega'^+}\int_{g^{-1}(t)}^\infty K\left(\frac{z-g^{-1}(X_i)}{h}\right)\mathrm{d}z\\
	&=&\frac{1}{n}\sum_{i=1}^n\int_{-\infty}^\infty K(u)\mathrm{d}u\\
	&=&1.
\end{eqnarray*}
For $\lim_{t\rightarrow\omega'^+}\widetilde{\mathbb{S}}_{X,1}(t)$, the use of the integration by subsitution and by parts means
\begin{eqnarray*}
	\lim_{t\rightarrow\omega'^+}\widetilde{\mathbb{S}}_{X,1}(t)&=&\frac{1}{n}\sum_{i=1}^n\lim_{t\rightarrow\omega'^+}\int_{g^{-1}(t)}^\infty g'(z)V\left(\frac{z-g^{-1}(X_i)}{h}\right)\mathrm{d}z\\
	&=&-\omega'+\frac{1}{n}\sum_{i=1}^n\int_{-\infty}^\infty g(g^{-1}(X_i)+hu)K(u)\mathrm{d}u\\
	&=&\frac{1}{n}\sum_{i=1}^n\int_{-\infty}^\infty[g(g^{-1}(X_i))+hg'(g^{-1}(X_i))u+O_p(h^2)]K(u)\mathrm{d}u-\omega'\\
	&=&\bar{X}-\omega'+O_p(h^2).
\end{eqnarray*}
On the other hand, the fact $\lim_{x\rightarrow-\infty}V(x)=1$ results in
\begin{gather*}
\lim_{t\rightarrow\omega'^+}\widetilde{S}_{X,2}(t)=\frac{1}{n}\sum_{i=1}^n\lim_{t\rightarrow\omega'^+}V\left(\frac{g^{-1}(t)-g^{-1}(X_i)}{h}\right)=1,
\end{gather*}
and
\begin{eqnarray*}
	\lim_{t\rightarrow\omega'^+}\widetilde{\mathbb{S}}_{X,2}(t)&=&\frac{1}{n}\sum_{i=1}^n\int_{-\infty}^{g^{-1}(X_i)}g'(z)\lim_{t\rightarrow\omega'^+}V\left(\frac{g^{-1}(t)-z}{h}\right)\mathrm{d}z\\
	&=&\frac{1}{n}\sum_{i=1}^n\int_{-\infty}^{g^{-1}(X_i)}g'(z)\mathrm{d}z\\
	&=&\bar{X}-\omega'.
\end{eqnarray*}
Then, the theorem is proven.
\end{document}